\documentclass[aps,pra,showpacs,amsmath,amssymb,superscriptaddress,reprint,10pt,longbibliography,noeprint]{revtex4-1}

\usepackage{amsmath,amsfonts,amssymb,amsthm,bbm,graphicx,enumerate,times}
\usepackage{mathtools}
\usepackage[usenames,dvipsnames]{color}
\usepackage{todonotes} 
\usepackage{comment}
\usepackage{hyperref}
\usepackage{tikz}
\usepackage{graphicx}
\usepackage{float}
\usetikzlibrary{shapes,positioning}
\usepackage{etoolbox}

 \definecolor{jens}{rgb}{.1,0.5,.4}

 \newtheorem{theorem}{Theorem}

 \newtheorem{lemma}[theorem]{Lemma}

 \newtheorem{observation}[theorem]{Observation}

\definecolor{henrik}{rgb}{.8,.3,0}

\newcommand{\mc}[1]{\mathcal{#1}}

\newcommand{\mb}[1]{\mathbb{#1}}

\newcommand{\e}{\mathrm{e}}

\newcommand{\tr}{\mathrm{Tr}} 
\newcommand{\Tr}{\mathrm{Tr}} 

\newcommand{\id}{\mb{1}}

\newcommand{\one}{\mathbb{I}}
\newcommand{\1}{\mathrm{id}}

\newcommand{\RR}{\mb{R}}

\renewcommand{\1}{\id}

\newcommand{\norm}[1]{\left\Vert #1 \right\Vert}

\newcommand{\ket}[1]{\left.\left|{#1}\right.\right\rangle}

\newcommand{\bra}[1]{\left.\left\langle{#1}\right.\right|}

\newcommand{\braket}[2]{\left\langle #1 \middle| #2 \right\rangle}

\newcommand{\ketbra}[2]{\ket{#1} \!\! \bra{#2}}

  \newcommand{\proj}[1]{\ketbra{#1}{#1}}

\begin{document}

\title{Single-shot holographic compression from the area law}
 
\author{H.\ Wilming}

\affiliation{Institute  for  Theoretical  Physics,  ETH  Zurich,  8093  Zurich,  Switzerland}

\author{J.\ Eisert}
\affiliation{Dahlem Center for Complex Quantum Systems, Physics Department, Freie Universit{\"a}t Berlin, 14195 Berlin, Germany}
\affiliation{Department of Mathematics and Computer Science, Freie Universit{\"a}t Berlin, 14195 Berlin, Germany}

\begin{abstract}
The area law conjecture states that the entanglement entropy of a region of space in the ground state of a gapped, local Hamiltonian only grows like the surface area of the region.
We show that, for any state that fulfills an area law, the reduced quantum state of a region of space can be unitarily compressed into a thickened boundary of the region. If the interior of the region is lost after this compression, the full quantum state can be recovered to high precision by a quantum channel only acting on the thickened boundary.
The thickness of the boundary scales inversely proportional to the error for arbitrary spin systems and logarithmically with the error for quasi-free bosonic systems. 
Our results can be interpreted as a single-shot operational interpretation of the area law. 
The result for spin systems follows from a simple inequality showing that any probability distribution with entropy $S$ can be approximated to error $\varepsilon$ by a distribution with support of size $\exp(S/\varepsilon)$, which we believe to be of independent interest.
We also discuss an emergent approximate correspondence between bulk and boundary operators and the relation of our results to tensor network states.
\end{abstract}
\maketitle

The \emph{area law conjecture} arguably constitutes one of the key insights in many-body physics in the last decades. 
It states that the entanglement entropy $S$ of a spatial region $A$
in a pure ground state of a gapped system with local interactions is 
at most proportional to the surface area $|\partial A|$ of the region,
\begin{align}\label{eq:arealaw}
	S(A)  \leq k\, | \partial A |,
\end{align}
where $k>0$ is some constant independent of $A$ \cite{Srednicki,AreaReview}. 
It has been rigorously proven for large classes of systems: for lattice spin-models with sufficiently well-behaved low-energy density of states \cite{Masanes2009, Brandao2015}, 
for one-dimensional gapped lattice spin-models \cite{HastingsAreaLaw,VaziraniAreaLaw,Brandao2013} and for
bosonic non-interacting gapped lattice models in any dimension \cite{Area,GeneralBosonicAreaLaw,NJPAreaLaw}.
In gapless systems, the area law is modified by a logarithmic correction for non-interacting fermions 
\cite{WolfAreaLaw, PhysRevLett.105.050502,PhysRevB.92.115129,Halfspace}
and more generally for systems described by a $1+1$ dimensional conformal field theory \cite{Holzhey1994,Calabrese2004}.
In contrast to states that fulfill an area law, the entanglement entropy of a region scales like the volume of the region instead of its surface area for generic pure states \cite{Lubkin1978,Lloyd1988,Page1993}. The area law therefore tells us that the corresponding states are little entangled compared to generic states.
This property lies at the heart of modern ansatz classes for ground states of many-body systems in terms of tensor network states, many of which fulfill an area law by construction 
\cite{Orus-AnnPhys-2014,AreaReview,VerstraeteBig,SchuchReview,PhysRevLett.105.010502}.

Nevertheless,  the operational meaning of the area law for a single system is not at all obvious from the point of view of quantum information theory.
This is because the entanglement entropy quantifies the ratio of maximally entangled bits that can be distilled from many identical and independent copies of the system \cite{Bennett1996} and therefore does not quantify in a meaningful way the entanglement content of a single system. 
In many-body physics, however, we are usually concerned with a large, but single system.  
It is therefore desirable to find a clear interpretation of the area law in terms of the von~Neumann entropy that pertains to single systems
\footnote{This is a complementary problem to proving area laws in terms of single-copy quantifiers of entanglement, which quantify how many maximally entangled bits can be distilled from a single copy of the system, as advocated in Ref.~\cite{Cramer2005}.}.

 \begin{figure}
\includegraphics[width=0.48\textwidth]{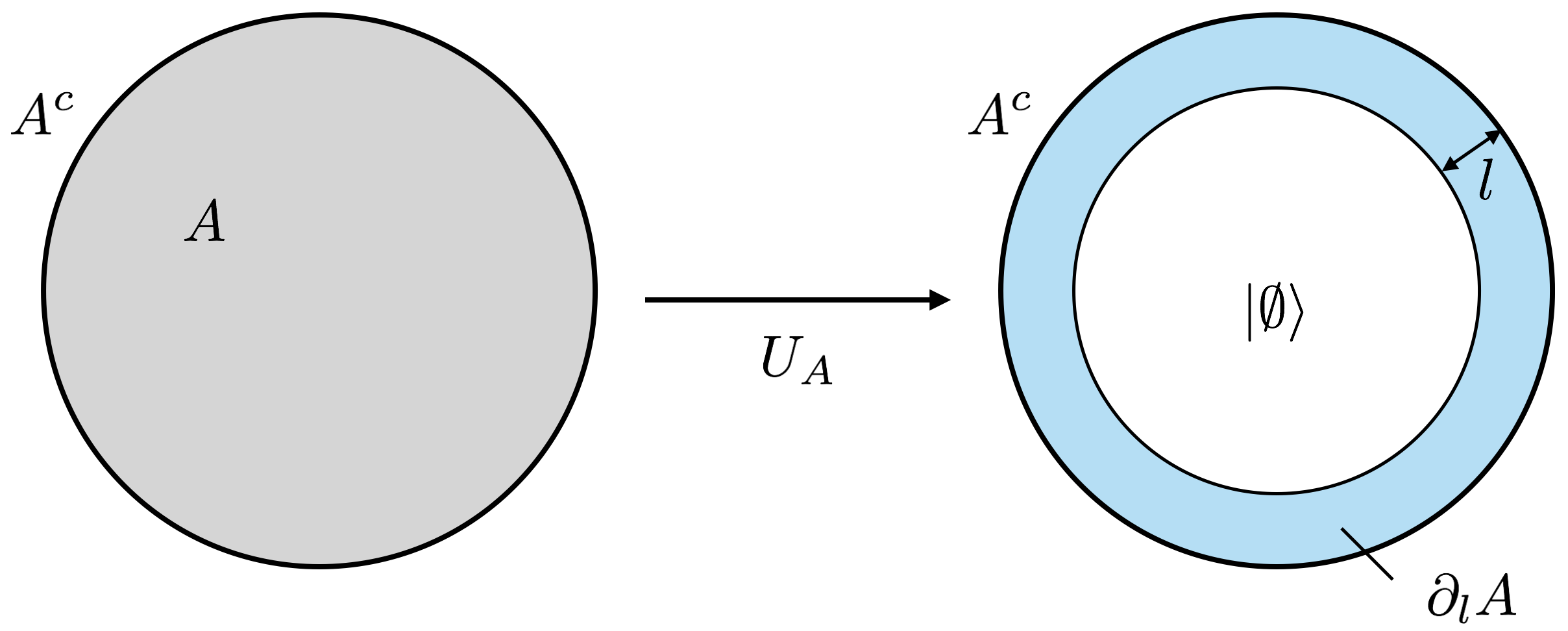}
	 \caption{\textbf{Holographic compression:} The unitary $U_A$ approximately compresses the quantum information in a region $A$
	 into a thickened boundary $\partial_l A$ of thickness $l$ and replaces the quantum state in the interior by a reference state vector $\ket{\emptyset}$ that can be chosen freely. 
	 The thickness $l$ required for an approximation error $\varepsilon$ grows as $1/\varepsilon$ with the error for spin systems and as $\log(1/\varepsilon)$ for quasi-free bosonic lattice models. 
	 It is always independent of the system size and independent of the size of $A$ for spin systems. For bosonic systems it depends very weakly on $A$ as $\log(|A|)$.}
	\label{fig:compression}
 \end{figure}

In this work, we provide such an interpretation by showing that for any pure state vector fulfilling an area law the quantum information contained in a region of space can be approximately compressed into a thickened boundary of the region by a unitary operation acting only within the region. This clarifies the quantitative meaning of the area law and, as we show below, brings it in contact with current research in the field of quantum error correction and holography.
 
Indeed, the area law is sometimes compared to the holographic principle in fundamental physics, which, very roughly speaking, states that the entropy contained in a region of space is bounded by its surface area \cite{RevModPhys.74.825}.
The most well-known example for this behavior is the Bekenstein-Hawking formula for a Schwarzschild black hole, which associates an entropy proportional to the surface area of the event horizon to a black hole \cite{PhysRevD.7.2333,Hawking,PhysRevLett.96.181602}.
The holographic principle suggests that it is possible to explicitly encode the quantum state of a region of space into one ``on'' its boundary. 
This idea is made manifest in the celebrated AdS/CFT correspondence, where quantum gravity on an 
Anti-de Sitter (AdS) background is described by a conformal field theory on one lower spatial dimension \cite{Maldacena98,Witten}.
Our results show that a similar encoding exists, at least approximately, for any pure quantum state that obeys an area law. We therefore call the corresponding procedure \emph{holographic compression}.

\emph{Main result. } 
To state our results we need to introduce some notation. 
We assume a quantum system is defined on some regular $D$-dimensional lattice $\Lambda$. To each site  $x\in\Lambda$ we associate a Hilbert-space $\mc H_x$ with $\mathrm{dim}(\mc H_x)= d$. 
Later, we also present results for harmonic, i.e., non-interacting bosonic, lattice systems, which can be seen as lattice-regularized free field theories.
For any region $A\subset \Lambda$ and any distance $l$ we introduce the \emph{thickened boundary} $\partial_l A$ as the set of sites in $A$ whose distance to the complement $A^c=\Lambda\setminus A$ of $A$ is less than $l$ in the lattice-distance (see Fig.~\ref{fig:compression}). We denote by $|A|$ the number of lattice sites in $A$.
In correspondence to the thickened boundary we also define the \emph{bulk} of $A$ as $A\setminus \partial_l A$. 
Finally, to state our result precisely, we define the function 
\begin{align}
	l_A(k) \coloneqq \min \left\{ l : |\partial_{l} A|\geq k|\partial A| \right\}.
\end{align}
The quantity $l_A(k)$ is the minimum distance such that the thickened boundary of width $l_A(k)$ contains at least $k |\partial A|$ sites. 
For large, smooth regions, such as spheres, we have $l_A(k)\sim k$.  
Finally, we write $\approx_\varepsilon$ to indicate that an equation holds up to an error $\varepsilon$ in terms of fidelity \cite{Nielsen2000}, i.e., for pure states $\ket{\psi} \approx_\varepsilon \ket{\phi}$ iff $|\langle\psi|\phi\rangle |^2 \geq 1-\varepsilon$. 
\begin{theorem}[Holographic compression]\label{thm:main}
	Consider a quantum state $\rho=|\psi\rangle\langle \psi|$ on the lattice fulfilling an area law, $S(A)\leq k|\partial A|$.
	Then for any region $A$, any $\varepsilon>0$ and any $l\geq l_A(k/\varepsilon)$, there exists a unitary $U_A$ acting only in $A$ which disentangles 
the bulk of $A$ from the complement of $A$,
\begin{align}\label{eq:basicequation}
 U_A   \ket{\psi}    \approx_\varepsilon|\chi \rangle\otimes \ket{\emptyset}, 
\end{align}
	where $|\chi\rangle$   is supported on $A^c\cup \partial_{l}A$ 
	and $\ket{\emptyset}$ is an arbitrary state vector on the bulk of $A$ that can be chosen freely. 
\end{theorem}
The situation is illustrated in Fig.~\ref{fig:compression}. 
The bulk of $A$ contains all of the lattice sites of $A$ up to a fraction 
$\sim k |\partial A| / ( \varepsilon |A|)$.
If $A$ is a sphere of radius $r$, then only a fraction $\sim k/(\varepsilon r)$ of sites remain entangled with the outside after $U_A$ is applied. 
On large scales, this fraction is arbitrarily small.
Yet, effectively all of the quantum information contained in region $A$ is compressed into the thickened boundary. 
If after the compression the bulk is lost, we can simply recover the global state vector $\ket{\psi}$ to error $\varepsilon$ by first tensoring with the fixed state vector $\ket{\emptyset}$ and then applying $U_A^\dagger$. 

We can also use holographic compression on the complement of the region $A$, compressing all the quantum information contained outside of $A$ into a thickened boundary of $A$ outside of $A$. 
This effectively produces a state vector $|\chi\rangle_{\overline{A}}$, supported on $A$ and a small boundary, whose reduced state coincides with the original state on $A$ up to a small error.
Such a \emph{holographic purification} implies that local expectation values and measurement statistics in region $A$ can in principle be computed up to a small error using pure states defined on a region only slightly larger than $A$ \footnote{This is indeed what happens in tensor network states with constant bond dimension, where all of the above is true exactly (with $\varepsilon=0$).}. 
Of course, instead of the thickened boundary, we could also compress the information in $A$ into any other subregion of $A$ that contains at least $S(A)/\varepsilon$ many sites.
If $\ket{\psi}$ arises as a lattice-description of a $D+1$-dimensional quantum field theory (QFT) with lattice distance $a$, the area law should be formulated as $S(A) \leq k/a^{D-1} |\partial A|$ to include the dependence on the lattice spacing. In this case the thickness of the thickened boundary scales as $l\sim k a$ and hence \emph{decreases} as we send the lattice-spacing to zero \footnote{note3}. 
A rigorous formulation of holographic compression and a proof of a statement similar to Thm.~\ref{thm:main} in the continuum would likely have to proceed via operator algebraic methods instead of the elementary methods we present here \cite{Haag1996,Witten2018}. 
We hence leave it to future work.

\emph{The proof. } Let us now explain how to arrive at our main result and then discuss its consequences in more detail. 
We consider a state vector $\ket{\psi}$ that fulfills an area law as in \eqref{eq:arealaw}, where the von-Neumann entanglement entropy is given by
$S(A) = -\tr\left(\rho_A \log \rho_A\right)$, with $\rho_A = \Tr_{A^c}(\proj{\psi})$.
We assume that logarithms are taken with base $d$, i.e., $\log(d)=1$.  
Then the maximum amount of entropy in a region $A$ is given by $\log(d^{|A|})=|A|$. 
For any region $A$, the state vector $\ket{\psi}$ can always be written, using a Schmidt-decomposition, 
as 
\begin{align}\label{eq:schmidt}
	\ket{\psi} = \sum_{j=1}^{d^{|A|}} \sqrt{p_j} \ket{j}_A \otimes \ket{j}_{A^c},
\end{align}
where $\{\ket{j}_A\}_j$ is an orthonormal basis on $\mc H_A=\otimes_{x\in A}\mc H_x$ (and similarly $A^c$). 
The numbers $0\leq p_j\leq 1$ correspond to the spectrum of $\rho_A$. 
We gather them in a vector $\mathbf{p}$ and order them non-increasingly $p_1 \geq p_2 \geq \cdots$. 
The entanglement entropy of $A$ is then given by their Shannon entropy:
\begin{align}
	S(A) = H(\mathbf{p}) \coloneqq -\sum_j p_j \log(p_j).
\end{align}
Ideally, our goal is to find a unitary $U_A$ that brings the reduced state on $A$ into the form
\begin{align}\label{eq:ideal}
	U_A\rho_AU_A^\dagger = \tilde \rho_{\partial_l A} \otimes \proj{\emptyset}, 
\end{align}
where $\tilde \rho_{\partial_l A}$ is supported on a thickened boundary of width $l$ (as small as possible) and $\ket{\emptyset}$ is some pure state on the corresponding bulk. 
Suppose for a moment that $p_j = 0$ for $j > d^{N}$ with $N\leq |A|$, i.e., $\rho_A$ has rank smaller or equal to $d^N$. 
In this case, we could always divide $A$ into a region $A_1$ containing $N$ sites and the rest $A_2$. 
Now let $\{\ket{\tilde j}_{A_1}\}$ be an arbitrary basis in the Hilbert-space on $A_1$ and $\ket{\emptyset}_{A_2}$ an arbitrary pure state on $A_2$.
Then we could always find a unitary $U_A$ such that \footnote{Note4}
\begin{align}
	U_A \ket{j}_A = \ket{\tilde j}_{A_1} \otimes \ket{\emptyset}_{A_2},\quad j=1,\ldots, d^N
\end{align}
and 
\begin{align}
	U_A \rho_A U_A^\dagger = \sum_{j=1}^{d^N} p_j |\tilde j\rangle\langle \tilde j|_{A_1} \otimes \proj{\emptyset}_{A_2} = \tilde \rho_{A_1} \otimes \proj{\emptyset}_{A_2}.\nonumber
\end{align}
By choosing $U_A$ appropriately we could hence compress all the quantum information contained in $A$ into a subsystem of size $N$. 
In particular if $N\leq K|\partial A|$ for some constant $K>0$, we could compress all the quantum information into a thickened boundary of width $l\sim K$. 
Unfortunately, however, in general the state $\rho_A$ has full rank, i.e., $p_j>0$ for all $j$. 
This is not in contradiction with $S(A)$ being small: even a state with infinite rank can have arbitrarily small entropy. 
Nevertheless, we now show that the above procedure works approximately if we consider instead of $\ket{\psi}$ an approximation of it, where we only keep the $M$ largest $p_j$ and choose $M$ correctly. 
Let $P_M$ be the projector onto the subspace spanned by $\{\ket{j}_A\}_{j=1}^M$ and consider the state vector
\begin{align}\label{eq:stateapprox}
	\ket{\psi_M} &\coloneqq  \frac{P_M \otimes \one \ket{\psi}}{{\bra{\psi} P_M\otimes \one \ket{\psi}^{1/2}}} \\
			    &= \frac{1}{({\sum_{k=1}^M p_k})^{1/2}} \sum_{j=1}^M \sqrt{p_j} \ket{j}_A\otimes \ket{j}_{A^c},\nonumber
\end{align}
where we have used that 
\begin{align}
	|\braket{\psi}{\psi_M}|^2 = \bra{\psi}P_M\otimes \one \ket{\psi} = \sum_{j=1}^M p_j. 
\end{align}
The following lemma shows that $\ket{\psi_M}$ is a good approximation to $\ket{\psi}$ if $M$ is sufficiently large (see Appendix~\ref{app:lemma2} for a proof). 
\begin{lemma}[Low-rank approximation of low-entropy distributions]\label{lemma:main}
	Let $\mathbf{p}=\mathbf{q}\oplus \mathbf{r}$ be any ordered probability distribution, where $\mathbf{q}$ contains the largest $M$ entries and $\mathbf{r}$ the remaining entries of $\mathbf{p}$. 	Then 
	\begin{align}
		\sum_{j=1}^M p_j \geq 1 - \frac{H(\mathbf{p})-H(\mathbf{q})}{\log(M) - \log(\sum_{k=1}^Mp_k)} \geq 1 - \frac{H(\mathbf{p})}{\log(M)}. 
	\end{align}
	\end{lemma}
We emphasize that this result requires no assumption about the distribution $\mathbf{p}$ and also applies to infinite-dimensional discrete probability distributions. 
It can be seen as a single-shot compression argument. 
Due to the monotonicity of R\'enyi entropies, the above result implies a similar result for all R\'enyi entropies with $\alpha<1$. 
However, using an explicit counter-example we show the following in Appendix~\ref{app:counterexample}:

\begin{observation}[Higher R\'enyi entropies] Holographic compression cannot be
deduced from an area law for R\'enyi entropies with parameter $\alpha>1$.
\end{observation}

As a further side-remark, let us mention that Lemma~\ref{lemma:main} can be reformulated in terms of \emph{smooth entropies}, appearing, e.g., in quantum cryptography, as $\varepsilon H^\varepsilon_0(\mathbf{p}) \leq H(\mathbf{p})$, where $H^\varepsilon_0$ denotes the smoothed max entropy \cite{Smooth}. 
Coming back to our application, let us now choose $M$ and $l$ such that 
\begin{align}\label{eq:conditionM}
	|\partial_l A|\geq \log(M) \geq \frac{k}{\varepsilon} |\partial A| \geq \frac{S(A)}{\varepsilon}.
\end{align}
For example, we could choose $l = l_A(k/\varepsilon)$ and $M$ as the dimension of the full Hilbert-space of the corresponding thickened boundary. 
For any choice fulfilling~\eqref{eq:conditionM} we have $|\braket{\psi}{\psi_M}|^2 \geq 1 - \varepsilon$.
We can now define the unitary $U_A$ with respect to the state vector $\ket{\psi_M}$ as before, i.e., as any unitary that fulfills
\begin{eqnarray}\label{eq:definingunitary}
	U_A P_M U_A^\dagger = \tilde P_M\otimes \proj{\emptyset}   \coloneqq \sum_{j=1}^M |\tilde j\rangle\langle\tilde j|\otimes |\emptyset\rangle\langle\emptyset|   , 
\end{eqnarray}
where $\tilde P_M$ is any $M$-dimensional subspace on the thickened boundary $\partial_{l} A$ with basis $\{\ket{\tilde{j}}\}$.
For any unitary fulfilling this condition, we then obtain \eqref{eq:basicequation} by setting
\begin{align}
	|\chi \rangle\otimes\ket{\emptyset} \coloneqq U_A \ket{\psi_M}, 
\end{align}
which completes the proof of the theorem,
since
\begin{align}
	|\langle \psi| U_A^\dagger |\chi\rangle\otimes|\emptyset\rangle|^2 = |\langle\psi |\psi_M\rangle|^2 \geq 1 - \varepsilon. 
\end{align}

\emph{Bulk-boundary correspondence. } 
Condition \eqref{eq:definingunitary} is reminiscent of the bulk-boundary correspondence in holography, where a local operator supported in the bulk of space can be represented as an, in general non-local, operator on the corresponding boundary description. 
In our case, any operator supported within the subspace $P_M$ is mapped exactly to the thickened boundary. More generally, consider $k$ operators $\{X_i\}_{i=1}^k$ supported in $A$ that commute with $P_M$.
Then the unitary $U_A$ acts on these operators as
\begin{align}
	U_A X_i U_A^\dagger  = \tilde X_i\otimes \proj{\emptyset} + \tilde Q_M \tilde X'_i \tilde Q_M,
\end{align}
where $\tilde Q_M = U_A(\one - P_M)U_A^\dagger$. It follows that
\begin{align}
	\langle \psi| X_1 \cdots X_k|\psi \rangle \approx \tr\left(\tilde X_1 \cdots \tilde X_k \tilde \rho_{\partial_l A}\right),
\end{align}
i.e., correlation functions of operators that fulfill $[X,P_M]=0$ can equivalently be computed (with small error) in the representation on the thickened boundary.
The subspace $P_M$ takes a similar role as a low-energy subspace or code subspace in models of 
holography in terms of quantum error correcting codes \cite{Harlow,Pastawski2015,Yang2016,Hayden2016,Pastawski2017,Harlow2017,Jahn}. 
Indeed, we can understand $P_M$ as the low-energy subspace of the \emph{entanglement} or \emph{modular} Hamiltonian $K_A \coloneqq -\log(\rho_A)$.  
It would be interesting to see in concrete physical models whether local operators are well approximated by operators that commute with $P_M$, making it possible to compute correlation functions directly in the representation on the boundary. 
By taking advantage of the freedom in choosing $U_A$ (the choice of basis $\{\ket{\tilde j}\}$), one could further try to optimize $U_A$ in such a way that the boundary description is robust against loss of local regions on the boundary as in AdS/CFT. We leave these topics to future work.  

\emph{From entropy to energy. } We have shown that by taking $M=d^{\lceil S(A)/\varepsilon\rceil}$, the state vector $\ket{\psi_M}$ approximates $\ket{\psi}$ to accuracy $\varepsilon>0$. 
This automatically implies that they have similar energy \emph{density}, since 
\begin{align}
	\left|\langle \psi_M | H|\psi_M\rangle - \langle \psi | H|\psi\rangle\right|\leq 2\sqrt{\varepsilon} \norm{H},
\end{align}
where $\norm{H}$ denotes the operator norm of $H$, which scales extensively for local, bounded Hamiltonians. 
Nevertheless, the total energy difference might still diverge linearly with the system size $|\Lambda|$.
We now show a stronger result: 
The total energy in $\ket{\psi}$ and $\ket{\psi_M}$ differ by a constant (if $A$ is held fixed) and therefore the energy densities of $\ket{\psi}$ and $\ket{\psi_M}$ differ by an amount of order $1/|\Lambda|$. 
The proof of the following Theorem is provided in Appendix~\ref{app:energy}.

\begin{theorem}[Energetic area law]\label{thm:energy} Let $\ket{\psi_M}$ be defined as before with $M\geq d^{\lceil S(A)/\varepsilon\rceil}$. If $H$ is a uniformly bounded, local Hamiltonian with ground state vector $\ket{\psi}$, there exists a constant $h$ such that
	\begin{align}
		\left|\langle \psi_M | H|\psi_M\rangle - \langle \psi | H|\psi\rangle\right|\leq \left({\frac{\varepsilon }{1-\varepsilon}}\right)^{1/2} h\,
		|\partial A|. 
	\end{align}
\end{theorem}

\emph{Relation to tensor network states. } 
In recent years, \emph{tensor network states} have received a lot of attention as variational ansatz states for 
ground states of gapped many-body systems 
\cite{Orus-AnnPhys-2014,AreaReview,VerstraeteBig,SchuchReview,PhysRevLett.105.010502}
as well as providing concrete toy models for the AdS/CFT correspondence 
\cite{PhysRevD.86.065007,Pastawski2015,Lee2016,Yang2016,Hayden2016,Pastawski2017,WaveletsMERA,Jahn}. 
One of the reasons that specifically \emph{projected entangled pair states  (PEPS)} and states from 
\emph{multiscale entanglement 
renormalization (MERA)} \cite{MERA1} in spatial dimensions beyond unity, which can be efficiently embedded in PEPS \cite{PhysRevLett.105.010502},
 are believed to be a good ansatz class for ground states of gapped many-body systems is that they automatically fulfill an area law. 
 In fact, they fulfill something significantly stronger, namely an area law in terms of the rank of reduced density matrices,
\begin{align}
	S_0(A)_{\rho^{\mathrm{PEPS}}} \leq k | \partial A |,
\end{align}
where $S_0(A)_\rho \coloneqq \log(\mathrm{rank}(\rho_A))$ is the $0$-R\'enyi entropy and $\rho^{\mathrm{PEPS}}$ is the density matrix of the PEPS.
This means that PEPS can be holographically compressed \emph{exactly} onto a thickened boundary, as in \eqref{eq:ideal}. 
Despite the sucess of PEPS, it is in fact known that in general an area law in terms of the von~Neumann entropy does \emph{not} suffice for a state to be well approximable by a PEPS \cite{1411.2995,Footnote1AreaLaw}. 
Our results show that even though one cannot in general find a good PEPS approximation, another property of PEPS -- the ability to compress the quantum information in some region into its boundary -- is indeed implied by an area law. 
Furthermore, we find that pure states with small Schmidt-rank can indeed in principle be used to approximate the ground state of a local many-body system to high precision whenever this system fulfills an area law. The key difference to the case of PEPS is that the approximation $\ket{\psi_M}$ has a small Schmidt-rank over a fixed bipartition of the system, whereas PEPS have small Schmidt-rank for \emph{any} bipartition of the system.

\emph{Holographic compression for non-interacting bosons. }
So far, all our results were applied to arbitrary spin lattice models. 
For ground states of \emph{local, gapped, harmonic 
Hamiltonian models}, holographic compression is particularly interesting, not the least because in this situation the area law has actually been proven in all generality \cite{Area,GeneralBosonicAreaLaw,NJPAreaLaw}. 
For such systems, the above statement still holds true, with a proof analogous to the previous one, since Lemma~\ref{lemma:main} also holds for infinite dimensional discrete probability distributions. The transformation that implements
the holographic compression, however, will in general not be a Gaussian operation \cite{Continuous,GaussianQuantumInfo}, which limits its physical interpretation. 
Interestingly, for such models holographic compression with Gaussian operations holds true as well, 
but now derived from the Hamiltonian model. 
Moreover, the width of the thickened boundary has a much weaker scaling with the error $\varepsilon$ when compared to the case of arbitrary spin systems. 
In the following, we specifically consider families of Hamiltonian models of the form
\begin{equation}\label{frmt}
H = \mathbf{P}  \mathbf{P}^T/2  + \mathbf{X} V \mathbf{X}^T ,
\end{equation}
with $\mathbf{X}=(X_1,\dots, X_{|\Lambda|})$ and $\mathbf{P}=(P_1,\dots, P_{|\Lambda|})$ denoting the coordinates and momenta and allowing for arbitrary finite-ranged couplings between the 
position coordinates (our results generalize to models with momentum-couplings). 
\begin{theorem}[Gaussian bosonic  holographic compression]\label{thm:bosonic} Consider ground states of gapped, finite-ranged Hamiltonians of the form
as in Eq.\ (\ref{frmt}) on a regular lattice. Then there exist constants $C_1,C_2>0$
so that for  any region $A$, any $\varepsilon>0$ and any width 
       \begin{align}
	       l\geq l_A\left(C_1 (\log\left(C_2/\varepsilon\right)+\frac{5}{4}\log(|A|))\right),
       \end{align}
there exists a Gaussian unitary $U_A$ acting only in $A$ which disentangles 
the bulk of $A$ from the complement of $A$,
\begin{align}
	U_A \ket{\psi} \approx_\varepsilon |\chi\rangle\otimes \ket{\emptyset}, 
\end{align}
where $|\chi\rangle$ is supported on $A^c\cup \partial_{l}A$ and $\ket{\emptyset}$ is the vacuum 
state in the bulk.
\end{theorem}
The proof of this result is presented in Appendix~\ref{app:gaussian}.  
While the width of the thickened boundary now very slightly depends on the size of the region $A$ through $\log(|A|)$, the dependence on the error $\varepsilon$ only scales as $\log(1/\varepsilon)$, which is a strong improvement compared to the general result on spin systems, where the error scales as $1/\varepsilon$. 
This suggests that also in the case of spin systems an improved scaling with $\varepsilon$ can be achieved if one incorporates more properties of ground states of local Hamiltonians than just the area law.

\emph{Conclusions. } We provided an operational single-shot interpretation of the area law by showing that the quantum information contained in a region of space can be compressed onto its thickened boundary, with a small error independent of the system-size. 
For spin systems, our results hold without any further assumption on the state and Hamiltonian besides the fact that the state fulfills an area law.
The results follow from an inequality that shows that any probability distribution with Shannon entropy $H$ can be approximated to precision $\varepsilon$ by one with support of dimension $\exp(H/\varepsilon)$. 
We believe that this inequality will be useful beyond the application considered here when low-entropic quantum states on high dimensional spaces have to be approximated.   
For spin systems we also discussed the emergence of an approximate correspondence between bulk and boundary operators by showing that correlation functions of bulk-operators that (approximately) commute with the projection onto the low-energy subspace of the entanglement Hamiltonian can be calculated in the representation on the thickened boundary. 

For non-interacting, gapped bosonic systems we were able to show that the compression unitary can be taken as a Gaussian unitary and that the width of the thickened boundary can be much smaller than in the generic case (proportional to $\log(1/\varepsilon)$ instead of $1/\varepsilon$, albeit with a logarithmic dependence on the size of the region). This shows that our general results can be improved by taking into account more concrete properties of the Hamiltonian model than just the area law. 
While the area law has been demonstrated to hold in large classes of models, it still lacks a general proof for arbitrary gapped lattice models.
Approximate holographic compression provides a physically relevant, but weaker notion of the area law, which might admit a proof in full generality. It is the hope that the present work invites such endeavors.

\emph{Acknowledgements. } H.~W. would like to thank Renato Renner and Joseph M. Renes for interesting discussions regarding Lemma~\ref{lemma:main}.
J.~E. thanks the ERC (TAQ), the DFG (CRC 183 project B01, EI 519/14-1, EI 519/14-1, EI 519/7-1), and the Templeton Foundation 
for funding. 
This work has also received funding from the European Union's Horizon 2020
		research and innovation programme under grant agreement No 817482 (PASQuanS).
H.~W. further acknowledges support from the Swiss National Science Foundation through SNSF project
No. 200020\_165843 and through the National Centre of Competence in Research \emph{Quantum Science and Technology} (QSIT). 
This research was supported in part by Perimeter Institute for Theoretical Physics. Research at Perimeter Institute is supported by the 
Government of Canada through the Department of Innovation, Science and Economic Development and by the Province of Ontario through the Ministry of Research and Innovation. 

\bibliographystyle{apsrev4-1}

%

\appendix
\section{Proof of Lemma~2}
\label{app:lemma2}
Here, we provide the proof of Lemma~2. 
Denote by $q_M$ the smallest entry of $\mathbf{q}$. Then $r_j\leq q_M$ for all $j$. In the following we will use several times that 
$x\mapsto -\log(x)$ is a positive, monotonously decreasing function on $[0,1]$ and that $H(\mathbf{p}) = H(\mathbf{q}) + H(\mathbf{r})$ and $\tr(\mathbf{p}) \coloneqq  \sum_i p_i = \tr(\mathbf{q}) + \tr(\mathbf{r}) =1$. First we lower bound the entropy of $\mathbf{r}$ as
\begin{align}
	H(\mathbf{r}) &= \sum_j r_j (-\log(r_j)) \geq \sum_j r_j (-\log(q_M))\\
				       &= -\tr(\mathbf{r})\log(q_M) = (1-\tr(\mathbf{q}))(-\log(q_M)).\nonumber
\end{align}
This gives 
\begin{equation}
	\sum_{i=1}^M p_i = \tr(\mathbf{q}) \geq 1 - \frac{H(\mathbf{r})}{-\log(q_M)} = 1 - \frac{H(\mathbf{p})-H(\mathbf{q})}{-\log(q_M)}.
\end{equation}
Using $1\geq \tr(\mathbf{q}) \geq M q_M$ and hence $-\log(q_M) \geq \log(M) - \log(\tr(\mathbf{q}))$ then yields the first inequality. 
The second inequality follows because both $H(\mathbf{q})$ and $- \log(\tr(\mathbf{q}))$ are positive.
It is interesting to acknowledge that in this derivation, one completely disregards the term $H(\mathbf{q})$ in the last step, but one still gets a bound sufficiently tight for the arguments for holographic compression to work.

\section{Counter-example for R\'enyi entropies with $\alpha>1$}
\label{app:counterexample}
Here, we show by example that an analogous result to Lemma~2 does not hold in general for any R\'enyi entropy with $\alpha>1$.
Recall that, given a distribution $\mathbf{p}$ on a space $X$, the classical R\'enyi entropies are defined by
\begin{align}
	H_\alpha(X) \coloneqq \frac{1}{1-\alpha}\log\left(\sum_{i\in X} p_i^\alpha \right) 
\end{align}
and the R\'enyi divergence $S_\alpha$ of a quantum state is given by the classical R\'enyi entropy of the spectrum of the quantum state.

In the following, we construct a specific probability distribution on $d^{|A|}$ events which fulfills
\begin{align}
	H_\alpha \leq k(\alpha) |\partial A|, \quad \forall \alpha>1,
\end{align}
and show that to approximate it to error $\varepsilon$ by one with rank $M$ requires $M$ to grow like $d^{|A|}$. 
Here and in the following, we drop the argument in the entropy function, since we always consider the entropy of the full distribution. 
We will use that the R\'enyi entropies for $\alpha>1$ fulfill the inequality \cite{Wilming2018}
\begin{align}
	H_\alpha \geq H_\infty \geq \frac{\alpha-1}{\alpha}H_\alpha.
\end{align}
This inequality shows that the scaling behavior of $H_\alpha$ (with system size) is independent of $\alpha$, up to a multiplicative factor that depends on $\alpha$. It hence suffices to study the case $\alpha\rightarrow\infty$, which is given by
\begin{align}
	H_\infty = - \log(p_{\mathrm{max}}),
\end{align}
where $p_{\mathrm{max}}$ is the largest entry of $\mathbf p$, or, in the quantum case, the largest eigenvalue of the corresponding density operator.
Now consider a quantum state on region $A$ with ordered spectrum
\begin{align}
	\mathbf p = \left(d^{-k |\partial A|}, \frac{1- d^{-k |\partial A|}}{d^{|A|}-1},\ldots,\frac{1- d^{-k |\partial A|}}{d^{|A|}-1}\right),
\end{align}
for some constant $k$. Then we have $H_\infty = k |\partial A|$ and hence
\begin{align}
	H_\alpha \leq \frac{\alpha}{\alpha-1}k |\partial A|,\quad \forall \alpha > 1.
\end{align}
However,
\begin{align}
	\sum_{i=1}^M p_i = d^{-k |\partial A|} + (M-1)\frac{1- d^{-k |\partial A|}}{d^{|A|}-1}.
\end{align}
If this quantity is supposed to be larger than $1-\varepsilon$, this requires
\begin{align}
	M \geq d^{|A|}\left[1 - \frac{\varepsilon}{1-d^{-k |\partial A|}}\right]\geq d^{|A|}(1-2\varepsilon),
\end{align}
where the last inequality holds for regions $A$ large enough such that $d^{-k|\partial A|}\leq 1/2$. We deduce that $\log(M)$ has to scale like the volume of $A$ and holographic compression cannot be achieved.

Here, we have simply written down a probability distribution by hand. It is not clear whether such probability distributions appear in realistic, 
local physical models. The purpose of this example was simply to demonstrate that if holographic compression is possible in systems which fulfill an area law for R\'enyi entropies with $\alpha>1$, then additional properties of the distribution have to come into play (for example, that they also fulfill an area law for $\alpha=1$).  

\section{Proof of Theorem~4}
\label{app:energy}
In this section, we provide the proof of Theorem~4. 
We split up the Hamiltonian as $H=H_A + H_{\partial A} + H_{A^c}$, where $H_A$ and $H_{A^c}$ contain the Hamiltonian terms that are supported fully within $A$ and $A^c$, respectively.
Since $H$ is a local, uniformly bounded Hamiltonian by assumption, we have $\norm{H_{\partial A}}\leq h |\partial A|$ for some constant $h$.
For the purpose of this proof, let us write $P_M\otimes\one$ simply as $P_M$ and let us assume without loss of generality that $H\ket{\psi}=0$.
Since $P_M$ is supported within $A$, we then have
\begin{align}
H P_M = [H_A + H_{\partial A},P_M] + P_M H.
\end{align}
This implies
\begin{align}
        \langle \psi_M |H|\psi_M\rangle &= \frac{\langle \psi | P_M H P_M|\psi\rangle}{\langle \psi|P_M|\psi\rangle}\nonumber\\
       &=\frac{\langle \psi | P_M [H_A + H_{\partial A},P_M]|\psi\rangle}{\langle \psi|P_M|\psi\rangle}.
\end{align}
Since $[\rho_A,P_M]=0$ and by the cyclicity of the trace we furthermore have
\begin{align}
        \langle\psi| P_M[H_A,P_M]|\psi\rangle = \tr(\rho_A P_M[H_A,P_M]) = 0.
\end{align}
We can therefore write
\begin{align}
        \langle \psi_M| H |\psi_M\rangle &= \frac{\langle \psi | P_M [H_{\partial A},P_M]|\psi\rangle}{\langle \psi|P_M|\psi\rangle}. \\
\end{align}
Using $P_M^2=\one$, we then have
\begin{align}
	\langle \psi | P_M [H_{\partial A},P_M]|\psi\rangle &\leq \norm{P_M \ket{\psi}}\norm{H_{\partial_A}(\ket{\psi}-P_M\ket{\psi})}\nonumber \\
                                                &\leq \norm{P_M\ket{\psi}} \norm{H_{\partial A}} \norm{\ket{\psi}-P_M\ket{\psi}}\nonumber \\
                                            &\leq \norm{P_M\ket{\psi}} h\,|\partial A| \sqrt{\epsilon}.
\end{align}
The result then follows using $\langle \psi|P_M|\psi\rangle = \norm{P_M\ket{\psi}}^2 \geq 1-\epsilon$.

\section{Gaussian holographic compression in bosonic systems}
\label{app:gaussian}
In this section, we provide the proof for Theorem~5.
On the lattice $\Lambda$, we consider ground states of families of \emph{gapped, finite-ranged, harmonic 
Hamiltonian models} of the form
\begin{equation}\label{fr}
H = \mathbf{P}  \mathbf{P}^T/2  + \mathbf{X} V \mathbf{X}^T ,
\end{equation}
with $\mathbf{X}=(X_1,\dots, X_{|\Lambda|})$ and $\mathbf{P}=(P_1,\dots, P_{|\Lambda|})$, allowing for arbitrary finite-range couplings between the 
position coordinates. Importantly, for such models, the existence of a gap is \emph{equivalent} to exponential decay of correlations and implies an area law \cite{NJPAreaLaw}.
The generalization to models also involving arbitrary couplings between momentum coordinates
is straightforward; for an in-depth discussion of such general harmonic models, see Ref.\ \cite{NJPAreaLaw}.

 The ground states of the harmonic models that we consider are pure Gaussian states. We will therefore pursue the discussion of holographic compression in these systems in terms of second moments
\cite{Continuous,GaussianQuantumInfo}. The second moments of $2m$ canonical coordinates $\mathbf{R}= (X_1,P_1,\dots, X_m, P_m)$ can be captured in terms of a real symmetric, positive definite covariance matrix with entries
\begin{equation}
	\gamma_{j,k} = {\rm tr}(\rho (R_j R_k+ R_k R_j) )
\end{equation}
for $j,k=1,\dots, 2m$.
We start by collecting a number of statements that will be made use of in the 
main argument. Particularly important will be the concept of a  symplectic eigenvalue.
An arbitrary covariance matrix
$\gamma$ of $m$ modes can always be brought into \emph{Williamson normal
form}, which means that there exists an $S\in Sp(2m,\RR)$ so that
\begin{equation}
	S \gamma S^T  = \text{diag} (d_1,d_1,\dots, d_m,d_m),
\end{equation}
with $d_1,\dots, d_m\geq 1$. Here, $Sp(2m,\RR)$ denotes the group of $2m \times 2m$ symplectic matrices, i.e., exactly those matrices that fulfill
\begin{align}
S\sigma S^T = \sigma,
\end{align}
where $\sigma$ is the symplectic form
\begin{equation}\label{sym}
	\sigma=\bigoplus_{j=1}^{m}
	\left[\begin{array}{cc}
	0 & 1\\
	-1 & 0
	\end{array}
	\right].
\end{equation}
The numbers $d_1,\dots, d_m$ are referred to as \emph{symplectic eigenvalues} of $\gamma$. 
In the following, we will need to order the $d_j$ in non-increasing order and non-decreasing order. 
We will therefore write $d_j^\uparrow$ for the $j$-th smallest symplectic eigenvalue and $d_j^\downarrow$ for the $j$-th largest symplectic eigenvalue. 
If the matrix in question needs to be specified explicitly, we give it as an argument to the symplectic eigenvalues, for example, 
$d^\downarrow_j(A)$ if the matrix in question is called $A$.  

The full covariance matrix $\gamma$ of the global pure Gaussian state defined on the lattice $\Lambda$ can be partitioned as
\begin{equation}\label{gamma}
	\gamma=\left[\begin{array}{cc}
	\gamma_A & \Xi\\
	\Xi^T & \gamma_{A^c}\end{array}\right], 
\end{equation}
according to sites in $A$ and $A^c$, respectively, $\Xi$ capturing correlations. By construction, the symplectic
spectrum of $\gamma$, so the eigenvalues of $|\sigma\gamma| $, are all given by $1$, as the quantum state is pure. 
As we will see, due to the strong correlation decay featured in $\Xi$, the symplectic spectrum of $\gamma$ will only be little disturbed by $\Xi$.
We will then use this fact to show that the ground state vector $\ket{\psi}$ can be brought close to a product vector by a local unitary.
The mentioned correlation decay is captured by the following Lemma, which follows from the results in Ref.\ \cite{NJPAreaLaw}.

\begin{lemma}[Exponential decay of covariance matrix entries]\label{ed} For all ground states of gapped, local harmonic models
of the form (\ref{fr}), the entries of $\Xi$ are exponentially decaying, in that there exist $c_1,c_2>0$
such that
\begin{equation}
	|\Xi_{j,k}|  \leq c_1 e^{-c_2 l}
\end{equation}
where $l={\rm dist}(j,k)$ is the lattice distance between sites $j$ and $k$.
$c_1, c_2$ depend on $\|V\|$, the range of the Hamiltonian, its spectral gap and the dimension of the lattice $\Lambda$ only.
\end{lemma}
We emphasize that the operator-norm $\norm{V}$ is independent of the system size due to the locality of the Hamiltonians that we consider. 
Let us now define the sets
\begin{align}
	A_{\leq l} \coloneqq \left\{ (j,k) \in (A\times A^c)\ |\ d(j,k)\leq l\right\}
\end{align}
and similarly the sets $A_{>l}$. 
Then we can write
\begin{align}
\Xi =	\Xi_{\leq l} + \Xi_{> l},
\end{align}
where $\Xi_{\leq l}$ is zero for all entries where $(j,k)\notin A_{\leq l}$ and $\Xi_{>l}$ containing the remaining entries. 
An important property of $\Xi_{\leq l}$ is that it is non-zero in at most $k_1 l |\partial A|$ rows and colums, where $k_1$ is some constant. 
Let us also define the truncated covariance matrices
\begin{align}
	\gamma_{\leq l} \coloneqq \begin{bmatrix} \gamma_A &\Xi_{\leq l} \\ \Xi^T_{\leq l} & \gamma_B\end{bmatrix}.
\end{align}
In a first step, we now relate the symplectic eigenvalues of $\gamma_A$ to those of $\gamma_{\leq l}$. 
To do that we use the \emph{interlacing theorem} for symplectic eigenvalues:
\begin{theorem}[Interlacing theorem \cite{Krbek2014,Bhatia2015,Marginal}] Let $A\in \RR^{2n\times 2n}$ be positive definite. 
	Partition $A$ as $A=[A_{i,j}]$ where each $A_{i,j}$, $i, j = 1, 2$, is an $n \times  n$ matrix. 
	A positive definite matrix $B \in \RR^{(2n-2)\times(2n-2)}$ is called an s-principal submatrix of $A$ if $B=[B_{i,j}]$, and each $B_{i,j}$ is an $(n-1) \times (n-1)$ principal submatrix of $A_{i,j}$ occupying the same position in $A_{i,j}$ for $i, j = 1, 2$. 
	In other words, $B$ is obtained from $A$ by deleting, for some $1 \leq i \leq n$, the $i$-th and $(n + i)$-th rows and columns of $A$. Then
	\begin{align}
		d^\uparrow_{j}(A) \leq d^\uparrow_{j}(B) \leq d^\uparrow_{j+2}(A),\quad 1\leq j\leq n-1,
	\end{align}
	where we adopt the convention that $d^\uparrow_{n+1}(A) = \infty$.
\end{theorem}
This theorem is adapted to the convention where covariance matrices are ordered such that the symplectic matrix takes the form 
\begin{align}
\sigma= \begin{bmatrix}0 &\one_n\\ -\one_n&0\end{bmatrix} 
\end{align}
instead of Eq.\ \eqref{sym}.
This corresponds to a basis change that exchanges 
\begin{align}
	{\mathbf{R}}=(X_1,P_1,\ldots,X_m,P_m)\mapsto (X_1,\ldots,X_m,P_1,\ldots,P_m).\nonumber
\end{align}
In our convention, the corresponding matrix $B$ is obtained by deleting two consecutive rows and the corresponding columns of $A$. 
In the interpretation of covariance matrices as quantum states this corresponds to tracing out the corresponding site from the lattice. 
We now use this theorem and erase all columns from $\gamma_{\leq l}$ on which $\Xi_{\leq l}$ is non-zero and all the rows in which the corresponding block of $\Xi^T_{\leq l}$ is non-zero. 
This amounts to tracing out $k_1 l |\partial A|$ many lattice sites from $A^c$ -- namely all those to which the sites in $A$ are correlated in the quantum state described by $\gamma_{\leq l}$. 
Let us denote the resulting covariance matrix by 
\begin{align}
	\tilde \gamma_{\leq l} = \begin{bmatrix} \gamma_A &0 \\0& \tilde\gamma_{A^c} \end{bmatrix},
\end{align}
which now describes a quantum state on a lattice with some sites removed, but the state on $A$ is unchanged. In particular, the set of symplectic eigenvalues of $\tilde \gamma_{\leq l}$ contains those of $\gamma_A$, due to the direct-sum structure.  
By interatively using the interlacing theorem, we now find that
\begin{align}
	d^\downarrow_{j}(\gamma_A)\leq d^\downarrow_{j}(\tilde \gamma_{\leq l}) \leq d^\downarrow_1(\gamma_{\leq l}),\quad j\geq 2 k_1 l |\partial A|.\label{eq:interlacing} 
\end{align}
In a second step, we now relate the symplectic eigenvalues of $\gamma_{\leq l}$ to those of $\gamma$ via the following perturbation bound.

\begin{theorem}[Perturbation bound \cite{Bhatia2015}] Let $\gamma,\gamma' \in \RR^{2n\times 2n}$ be two positive definite matrices and $\{d^\downarrow_j\}_{j=1}^{n},\{{d'}^\downarrow_j\}_{j=1}^{n}$ their ordered symplectic eigenvalues. Then
	\begin{align}
		\max_j |d^\downarrow_j - {d'}^\downarrow_j| \leq \left(\norm{\gamma}^{1/2} + \norm{\gamma'}^{1/2}\right)\norm{\gamma-\gamma'}^{1/2}. 
	\end{align}
\end{theorem}
We now use this theorem for the matrices $\gamma$ and $\gamma_{\leq l}$ using the following Lemma. 
\begin{lemma}[Off-diagonal perturbation] Let $\gamma$ be any covariance matrix fulfilling exponential decay of correlations as in Lemma~\ref{ed}. Then
\begin{align}
	\norm{\gamma - \gamma_{\leq l}} = \norm{\begin{bmatrix}0 & \Xi_{>l} \\ \Xi^T_{>l} &0 \end{bmatrix}} \leq k_2 \e^{- c_2 l}{|A|^{1/2}}, 
\end{align}
for some constant $k_2>0$.
	\begin{proof}
		The operator norm  is no larger than the Frobenius 2-norm,
		\begin{align}
			\norm{A}^2 \leq \norm{A}_F^2 =  {\tr(AA^T)} =   {\sum_{i,j} |A_{i,j}|^2}.
		\end{align}
		Therefore, we have
		\begin{align}
			\norm{\gamma- \gamma_{\leq l}}^2 \leq {2\sum_{j,k}(\Xi_{>l})_{j,k}^2}.
		\end{align}
		But then,
		\begin{align}
			\sum_{j,k} (\Xi_{>l})_{j,k}^2 &= \sum_{(j,k) \in A_{>l}} (\Xi_{>l})_{j,k}^2 \\
						&\leq \sum_{j\in A} \sum_{\substack{k\in A^c\\d(j,k)>l}} c_1^2 \e^{-2 c_2 d(j,k)}\\
						&\leq c_1^2 |A| \e^{-2c_2 l} \sum_{l'=0}^\infty\mathrm{poly}(l') \e^{-2 c_2 l'} \\&\leq  \frac{k_2^2}{2} \e^{-2c_2 l}|A|,\nonumber 
		\end{align}
		where we have decomposed the sum over $k$ into ``shells'' of constant width a distance $l'+l$ from $j$ away to obtain the second to last line.  
	\end{proof}
\end{lemma}
The exponential decay of correlations furthermore implies that both $\norm{\gamma}$ and $\norm{\gamma_{\leq l}}$ are bounded by a constant. This can be seen using the fact that the operator norm of a Hermitian matrix $A$ is upper bounded as
\begin{align}
	\norm{A} \leq \max_i \sum_j |A_{i,j}|.
\end{align}
Consequently,
\begin{align}
	d^\downarrow_1(\gamma_{\leq l}) \leq 1 + k_3 l \e^{-c_2 l/2}|A|^{1/4}.
\end{align}
Combining this with \eqref{eq:interlacing} and with the fact that the symplectic eigenvalues of $\gamma_A$ are a subset of those of $\tilde \gamma_{\leq l}$, we find
\begin{align}
	d^\downarrow_j(\gamma_A) \leq 1 + k_3 l \e^{-c_2 l/2}|A|^{1/4} ,\quad j \geq 2k_1 l |\partial A|. 
\end{align}
We will now order the symplectic eigenvalues in the following way. We imagine associating each symplectic eigenvalue with one lattice site inside of $A$, the larger it is, the closer to the boundary.
Indeed, this can be done exactly using a symplectic transformation in $A$ acting on the Schmidt-normal form of pure Gaussian states (see below). 
We will now fix some distance $L$ and sum up all symplectic eigenvalues with a distance at least $\sim L$ from the boundary, namely all symplectic eigenvalues but the largest $2 k_1 L |\partial A|$ ones. 
By performing the sum in terms of "shells" of sites of a fixed width with distance $l\geq L$ to the boundary of $A$, we find
\begin{align}
	\sum_{j \geq 2k_1 L |\partial A|} d^\downarrow_j(A) &\leq k_4 \e^{- c_2 L/2} |A|^{5/4} \sum_{l=0}^\infty \mathrm{poly}(l)\e^{-c_2 l/2}\\
	&\leq k_5 \e^{-c_2 L/2} |A|^{5/4}.
\end{align}
If we now define
\begin{align}
	M(L) = 2 k_1 L |\partial A|
\end{align}
and choose
\begin{align}
	L_\varepsilon \coloneqq \frac{2}{c_2}\log\left(\frac{|A|^{5/4} k_5}{\varepsilon}\right),
\end{align}
we find
\begin{align}
	\sum_{j\geq M(L_\varepsilon)} d^\downarrow_{j}(\gamma_A) -1 \leq \varepsilon, 
\end{align}
with
\begin{align}
	M(L_\varepsilon) = c_3 \left[\log\left(\frac{c_4}{\varepsilon}\right) + \frac{5}{4}\log\left( |A|\right)\right] |\partial A|, \label{eq:Mepsilon}
\end{align}
for some constants $c_3,c_4$. 
Given these results on the symplectic spectrum, we can now proceed to prove the actual result.
For any Gaussian pure state, there exists a Gaussian product unitary $U_A \otimes U_{A^c}$, which is represented by a symplectic matrix of the form $S_A\oplus S_{A^c}$ and brings the ground state into a normal form of two-mode squeezed states. On the level of covariance matrices, this is reflected as 
\begin{align}
	(S_A\oplus S_{A^c}) \gamma (S_A^T \oplus S_{A^c}^T) = \begin{bmatrix} D_A& E\\E^T & D_{A^c}^T \end{bmatrix},
\end{align}
where $D_A$ and $D_{A^c}$ take the form 
\begin{align}
D_A = \mathrm{diag}\left(d_1(\gamma_A),d_1(\gamma_A),\ldots, d_{|A|}(\gamma_A),d_{|A|}(\gamma_A)\right)
\end{align}
and similarly for $A^c$. 
Importantly, 
\begin{align}
d^\downarrow_j(\gamma_{A^c})=d^\downarrow_j(\gamma_A) 
\end{align}
for all $j\leq |A|$ and $d^\downarrow_j(\gamma_{A^c})=1$ for $j>|A|$, corresponding normal-modes in $A^c$ that are in the vacuum and hence uncorrelated with $A$, since their state is pure.  
In the following we always mean $d_j(\gamma_A)$ when we write $d_j$ and introduce the notation $\mu_j \coloneqq (d_j-1)^{1/2}$. 
The off-diagonal matrix $E$ then takes the form $E = \begin{bmatrix} F& 0 \end{bmatrix}$ with 
\begin{align}
	F = \mathrm{diag}(\mu_1,-\mu_1,\ldots,\mu_{|A|},-\mu_{|A|}).\nonumber 
\end{align}
The corresponding state vector can be understood as a tensor product of two-mode squeezed state vectors
\begin{align}
	(U_A \otimes U_{A^c}) |\psi\rangle &= \otimes_{j=1}^{|A|} |\tilde\psi_{d_j}\rangle \otimes |\emptyset\rangle \\
	&=\otimes_{j=1}^{|A|} \frac{2}{d_j+1} \sum_k e_j^k |k,k\rangle_j \otimes |\emptyset\rangle, 
\end{align}
where $|k,k\rangle_j$ denotes the number-basis of the $j$-th pair of normal-modes  and $|\emptyset\rangle$ the vacuum state 
vector on the remaining oscillators. 
Importantly, exactly one of the two oscillators in each $|\tilde\psi_{d_j}\rangle$ is contained in $A$ and one of them in $A^c$.  
For $d_j>1$, the numbers $e_j$ are defined through
\begin{align}
	e_j^2 = 1 - \left(\frac{2}{d_j+1}\right)^2.  
\end{align}
In the following, let us further write $|\tilde \psi\rangle \coloneqq U_A \ket{\psi}$ and, for any $M\leq |A|$, let us write $| \tilde \psi_M \rangle$ and $|\psi_M\rangle$ for the corresponding states that result when we replace all but the $M$ largest symplectic eigenvalues $d_j^\downarrow$ 
by $1$. 
This corresponds to replacing the corresponding two-mode squeezed states by pairs of oscillators in the vacuum.  
Explicitly, we have
\begin{align}
	U_{A^c} |\tilde \psi_M\rangle = (U_A \otimes U_{A^c} ) |\psi_M\rangle=  \otimes_{j=1}^M |\tilde\psi_{d^{\downarrow}_j}\rangle \otimes |\emptyset\rangle,
\end{align}
where $|\emptyset\rangle$ now denotes the vacuum state vector on all but $M$ mode-pairs, $|A|-M$ oscillators of which are located in $A$. 
Using the above formulas, one easily find that
\begin{align}
	|\langle \tilde \psi|\tilde \psi_M\rangle|^2  = \prod_{j=M+1}^{|A|}
	\left(\frac{2}{d_j+1}\right)^2.
\end{align}
Using this result we now get the following truncation formula.
\begin{lemma}[Truncation formula]\label{tf}
 For $M=1,\dots, |A|$,
\begin{equation}
	\ |\langle \tilde \psi| \tilde \psi_M\rangle |^2 \geq \exp\biggl(
	-\sum_{j=M+1}^{|A|} (d_j-1)
	\biggr).
\end{equation}
\end{lemma}
\begin{proof} We have
\begin{align}
	|\langle \tilde \psi| \tilde \psi_M\rangle |^2 &= \exp\left(- 2 \sum_{j=M+1}^{|A|} \log\left(\frac{d_j+1}{2}\right)\right)\\
	&\geq \exp\left(-2\sum_{j=M+1}^{|A|} \left(\frac{d_j-1}{2}\right)\right) \\&= \exp\left(-\sum_{j=M+1}^{|A|} \left(d_j -1\right)\right),
\end{align}
	where we have used the inequality $\log(x) \leq x-1$ for $x>0$. 
\end{proof}
We will now use the truncation formula together with $M(L_\varepsilon)$ from \eqref{eq:Mepsilon} to find
\begin{align}
	|\langle\tilde \psi| \tilde \psi_{M(L_\varepsilon)}\rangle|^2 \geq \exp(-\varepsilon) \geq 1-\varepsilon.
\end{align}
However, $|\tilde\psi_{M(L_\varepsilon)}\rangle $ is of the required product form, since all but $M(L_\varepsilon)$ of the sites in $A$ are in the pure vacuum state. This completes the proof.

\clearpage
\end{document}